\documentclass[12pt,a4paper,thmsb]{article}
\usepackage{eurosym}
\usepackage{amsmath}
\usepackage{amssymb}
\usepackage{amsfonts}
\usepackage{color}
\usepackage{amssymb}
\usepackage{pifont}

\newtheorem{theorem}{Theorem}

\newtheorem{proposition}{Proposition}
\newtheorem{corollary}{Corollary}

\newtheorem{remark}{Remark}
\newtheorem{example}{Example}

\newenvironment{proof}[1][Proof]{\noindent\textbf{#1.} }{\ \rule{0.5em}{0.5em}}

\begin{document}

\title{Anonymity in sharing the revenues from broadcasting sports
leagues\thanks{Financial support from Ministerio de Ciencia e Innovaci\'{o}n, MCIN/AEI/
10.13039/501100011033, through grants PID2020-113440GBI00 and
PID2020-115011GB-I00, Junta de Andaluc\'{\i}a through grant P18-FR-2933, and
Xunta de Galicia through grant ED431B2022/03 is gratefully acknowledged.}}
\author{\textbf{Gustavo Berganti\~{n}os}\thanks{%
ECOBAS, Universidade de Vigo, ECOSOT, 36310 Vigo, Espa\~{n}a. E-mail:
gbergant@uvigo.es.} \\
\textbf{Juan D. Moreno-Ternero}\thanks{%
Corresponding author. Department of Economics, Universidad Pablo de Olavide,
41013 Sevilla, Espa\~{n}a}}
\maketitle

\begin{abstract}
    We study the problem of sharing the revenues from broadcasting sports
    leagues axiomatically. Our key axiom is anonymity, the classical
    impartiality axiom. Other impartiality axioms already studied in
    these problems are equal treatment of equals, weak equal treatment of
    equals and symmetry. We study the relationship between all impartiality
    axioms. Besides we combine anonymity with other existing axioms in the literature. Some combinations give rise to new
    characterizations of well-known rules. The family of generalized split rules
    is characterized with anonymity, additivity and null team. The
    concede-and-divide rule is characterized with anonymity, additivity and
    essential team. Others combinations characterize new rules that had not been
    considered before. We provide three characterizations in which three axioms
    are the same (anonymity, additivity, and order preservation) the fourth one
    is different (maximum aspirations, weak upper bound, and non-negativity).
    Depending on the fourth axiom we obtain three different families of rules.
    In all of them concede-and-divide plays a central role.
\end{abstract}


\medskip {} 


\newpage

\section{Introduction}

The number of people paying to watch professional sports on television has
steadily increased in recent years. The amount of revenues raised from
selling broadcasting rights has consequently increased, to the extent of
becoming crucial for the management of sports organizations. Typically, the
sale is carried out collectively and then revenues are shared among
participating organizations. The revenue sharing rules differ worldwide. In
North American leagues, uniform sharing mostly prevails. In European
leagues, hybrid schemes involving lower bounds and performance-based rewards
are widespread.

In Berganti\~{n}os and Moreno-Ternero (2020a) we introduced what we dub the 
\textit{broadcasting problem}, a simple model in which the sharing process
is based on the broadcasting audiences. The model considers a double
round-robin tournament in which all games have a constant pay-per view fee.
Thus, the prior of the model is just a square matrix, whose entries indicate
the audience of the game involving the row team and the column team, at the
former's stadium. We have also studied this model theoretically in Berganti%
\~{n}os and Moreno-Ternero (2020b, 2021, 2022a, 2022b, 2022c, 2023a).
Besides, we apply it to La Liga (the Spanish Football League) in Berganti%
\~{n}os and Moreno-Ternero (2020a, 2021, 2023b).

The \textit{broadcasting problem} is an instance of resource allocation.
Other well known instances to which the literature has paid attention are:
airport problems (e.g., Littlechild and Owen, 1973), bankruptcy problems
(e.g., O'Neill, 1982; Thomson, 2019a), telecommunications problems (e.g.,
van den Nouweland et al. 1996), minimum cost spanning tree problems (e.g.,
Berganti\~{n}os and Vidal-Puga, 2021), transport problems (e.g., Algaba et
al. 2019, Esta\~{n} et al. 2021), inventory problems (e.g., Guardiola et
al., 2021), liability problems with rooted tree networks (e.g. Oishi et al.,
2023), knapsack problems (e.g., Arribillaga and Berganti\~{n}os 2023),
pooling games (Schlicher et al., 2020), m-attribute games (\"{O}zen et al.,
2022), urban consolidation centers (Hezarkhani et al., 2019), and scheduling
problems with delays (Gon\c{c}alves-Dosantos et al., 2020).

Two classical notions within the literature of resource allocation are 
\textit{additivity} and \textit{impartiality}. The former is a kind of
robustness axiom. Suppose that we can divide a problem in two small
problems. Two views are possible: we solve the original problem, or we solve
the small problems separately. \textit{Additivity} requires that these two
views produce the same outcome (e.g., Thomson, 2019b). In \textit{%
broadcasting problems}, \textit{additivity} says that revenues should be
additive on the audiences. As for \textit{impartiality}, it is a basic
requirement of justice which excludes ethically irrelevant aspects from the
allocation process. (e.g., Moreno-Ternero and Roemer, 2006). It is
frequently formalized in two ways. On the one hand, with the axiom of 
\textit{equal treatment of equals}, which says that if two agents are
equal/symmetric, according to some input of the problem, then both agents
should receive the same. On the other hand, with the axiom of \textit{%
anonymity}, which says that a permutation of the set of agents equally
permutes the allocation. This has an immediate implication: the name of the
agents becomes irrelevant in the allocation process. A rule that is partial
to some specific agent because of his name reflects an obvious unfairness in
practice. Thus, \textit{anonymity} is a normatively appealing axiom both in
theory and practice.

Three different versions of \textit{equal treatment of equals} have been
used in \textit{broadcasting problems}, depending on when two teams are
considered equal/symmetric. 
In the first one, precisely named \textit{equal treatment of equals}, two
teams are considered equal if each time they play a third team the audiences
are the same in both games (e.g., Berganti\~{n}os and Moreno-Ternero,
2020b). In the second one, named \textit{weak equal treatment of equals}, an
additional condition is imposed to consider two teams as equal: the
audiences in the two games played by the two teams (home and away) also
coincide (e.g., Berganti\~{n}os and Moreno-Ternero, 2022b). In the third
one, named \textit{symmetry}, two teams are considered equal simply when
their aggregate audiences in the league coincide (e.g., Berganti\~{n}os and
Moreno-Ternero, 2021). Thus, it is obviously the strongest of the three
axioms as it is the one imposing the least demanding condition for teams to
be equal/symmetric.\footnote{%
We could thus rename this axiom as \textit{strong equal treatment of equals}%
, to align with the names of the other two axioms, but we prefer to keep the
terms used in previous literature to ease the comparison with the results
from that literature.}

In this paper, we focus instead on \textit{anonymity} in broadcasting
problems. We first argue that \textit{anonymity} is stronger than \textit{%
weak \textit{equal treatment of equals}} but it is neither related to 
\textit{equal treatment of equals} nor to \textit{symmetry}. We then
consider its implications, when combined with other standard axioms in the
literature. As we shall show, this allows us to uncover the structure of
this rich model of \textit{broadcasting problems}.

Our starting point is its combination with \textit{additivity}. As we
recently showed in Berganti\~{n}os and Moreno-Ternero (2023a), the
combination of both axioms characterizes a family of general rules where the
amount received by each team has three parts: one depending on the home
audience of this team, another depending on the away audience of this team,
and the third part depending on the total audience of the tournament.

Two focal rules in \textit{broadcasting problems} are members of the family
just described: the so-called \textit{equal-split rule} and \textit{%
concede-and-divide}. In Berganti\~{n}os and Moreno-Ternero (2020a) we argued
that viewers of each game can essentially be divided in two categories:
those watching the game because they are fans of one of the teams playing
(called hard-core fans) and those watching the game because they think that
the specific combination of teams renders the game interesting (called
neutral fans). Besides, we argued that the revenue generated by hard-core
fans should be assigned to the corresponding team and the revenue generated
by neutral fans should be divided equally between both teams. The \textit{%
equal-split} rule and \textit{concede-and-divide} are extreme rules from the
point of view of treating fans. The \textit{equal-split} rule assumes that
only neutral fans exist, whereas \textit{concede-and-divide} assumes that
there are as many hard-core fans as possible (compatible with the
audiences). Then, \textit{equal-split} divides the audience of each game
equally among the two teams. \textit{Concede-and-divide} concedes to each
team (for each game) its number of hard-core fans and divides equally the
rest. Both rules are characterized by \textit{additivity}, \textit{equal
treatment of equals} and a third axiom: \textit{null team} for the former,
and \textit{essential team} for the latter (e.g., Berganti\~{n}os and
Moreno-Ternero, 2020a). We show in this paper that replacing \textit{equal
treatment of equals} by \textit{anonymity} in those results is a move with
striking consequences. More precisely, we show that the combination of 
\textit{anonymity} with \textit{additivity} and \textit{null team}
characterizes a family of rules where the audience of each game is divided
in proportions $\lambda $ and $1-\lambda $ for the home and the away team
respectively (the \textit{equal-split rule} corresponds to the case where $%
\lambda =0.5).$ We show, on the other hand, that the combination of \textit{%
anonymity} with \textit{additivity} and \textit{essential team} still
characterizes \textit{concede-and-divide}. 

Another focal rule satisfying \textit{anonymity} and \textit{additivity} is
the \textit{uniform rule} (the audience of the tournament is divided equally
among all teams). Convex combinations of each of the resulting pairs from
the three focal rules have been considered. The so-called family of $EC$
rules is made of convex combinations of the \textit{equal-split rule} and 
\textit{concede-and-divide}. The family of $UE$ rules is made of convex
combinations of the \textit{uniform rule} and the \textit{equal-split rule}
and the family of $UC$ rules is made of convex combinations of the \textit{%
uniform rule} and \textit{concede-and-divide}. The three families of rules
have been characterized in Berganti\~{n}os and Moreno-Ternero (2021, 2022a)
always involving one of the three axioms formalizing \textit{equal treatment
of equals}. We show in this paper that, if we consider \textit{anonymity}
instead, we derive three new families of rules, dubbed \textit{extended $EC$
rules}, \textit{extended $UE$ rules}, and \textit{extended $UC$ rules},
which contain, respectively, the three families mentioned above. 
To wit, \textit{extended $EC$ rules} are defined by two parameters $%
x^{\prime },y^{\prime }\in \left[ 0,1\right] .$ Each team $i$ receives an
allocation that is the sum of two parts. The first part is given by an $EC$
rule where $\lambda $ depends on $x^{\prime }$ and $y^{\prime }$. The second
part depends on $x^{\prime },$ $y^{\prime }$ and the allocation \textit{%
concede-and-divide} yields for an associated broadcasting problem. When $%
x^{\prime }\geq y^{\prime }$ the associated problem is obtained by
nullifying the home audiences of team $i$, whereas when $x^{\prime }\leq
y^{\prime }$ the associated problem is obtained by nullifying the away
audiences of team $i$. When $x^{\prime }=y^{\prime }$ the rule just becomes
an $EC$ rule. The other two extended families can be described in similar
ways. It turns out we obtain characterizations of the three families of
extended rules when combining \textit{additivity} and \textit{anonymity}
with other axioms from the literature.

We conclude the introduction acknowledging that our paper is related to some
other papers in which families of rules are also characterized using similar
axioms to those used in this paper. Some examples are the following. In
cooperative games with transferable utility, Casajus and Huettner (2013),
van den Brink et al. (2013), and Casajus and Yokote (2019) characterize the
family of values arising from the convex combination of the Shapley value
and the egalitarian value. In our setting, this would correspond to the
family of $UE$ rules. In minimum cost spanning tree problems, making use of 
\textit{additivity}, Trudeau (2014) characterizes the convex combination of
the so-called folk rule and the cycle-complete rule, whereas Berganti\~{n}os
and Lorenzo (2021) characterize the family of Kruskal sharing rules. In
bankruptcy problems, Thomson (2015a, 2015b) characterizes families of rules
satisfying the standard non-negativity condition.

The rest of the paper is organized as follows. In Section 2, we introduce
the model, rules and axioms. In Section 3, we obtain new characterizations
of some well-known families of rules. In Section 4, we obtain
characterizations of new families of rules. We conclude in Section 5.


\section{The model}

We consider the model introduced in Berganti\~{n}os and Moreno-Ternero
(2020a). Let $N$ be a finite set of teams. Its cardinality is denoted by $n$%
. We assume $n\geq 3$. For each pair of teams $i,j\in N$, we denote by $%
a_{ij}$ the broadcasting audience (number of viewers) for the game played by 
$i$ and $j$ at $i$'s stadium. We use the notational convention that $%
a_{ii}=0 $, for each $i\in N$. Let $A\in {\mathcal{A}}_{n\times n}$ denote
the resulting matrix of broadcasting audiences generated in the whole
tournament involving the teams within $N$.\footnote{%
We are therefore assuming a round-robin tournament in which each team plays
in turn against each other team twice: once home, another away. This is the
usual format of the main European football leagues. Our model could also be
extended to leagues in which some teams play other teams a different number
of times and play-offs at the end of the regular season, which is the usual
format of North American professional sports. In such a case, $a_{ij}$ is
the broadcasting audience in all games played by $i$ and $j$ at $i$'s
stadium.} As the set $N$ will be fixed throughout our analysis, we shall not
explicitly consider it in the description of each problem. Each matrix $A\in 
{\mathcal{A}}_{n\times n}$ with zero entries in the diagonal will thus
represent a \textit{problem} and we shall refer to the set of problems as $%
\mathcal{P}$.

Let $\alpha _{i}\left( A\right) $ denote the total audience achieved by team 
$i$, i.e., 
\begin{equation*}
\alpha _{i}\left( A\right) =\sum_{j\in N}(a_{ij}+a_{ji}).
\end{equation*}%
Without loss of generality, we normalize the revenue generated from each
viewer to $1$ (to be interpreted as the \textquotedblleft pay per
view\textquotedblright\ fee). Thus, we sometimes refer to $\alpha _{i}\left(
A\right) $ by the \textit{claim} of team $i$. When no confusion arises, we
write $\alpha _{i}$ instead of $\alpha _{i}\left( A\right) $.

For each $A\in {\mathcal{A}}_{n\times n}$, let $||A||$ denote the total
audience of the tournament. Namely, 
\begin{equation*}
||A||=\sum_{i,j\in N}a_{ij}=\frac{1}{2}\sum_{i\in N}\alpha _{i}.
\end{equation*}

\subsection{Rules}

A (sharing) \textbf{rule} $\left( R\right) $ is a mapping that associates
with each problem the list of the amounts teams get from the total revenue.
As we have mentioned above we normalize the revenue generated from each
viewer to $1$. Formally, $R:\mathcal{P}\rightarrow \mathbb{R}^{N}$ is such
that, for each $A\in \mathbf{\mathcal{P}}$, 
\begin{equation*}
\sum_{i\in N}R_{i}(A)=||A||.
\end{equation*}

We first introduce three focal rules that have been studied in Berganti\~{n}%
os and Moreno-Ternero (2020a; 2020b; 2021; 2022a; 2022b; 2022c; 2023c).

\bigskip

The \textit{uniform rule} divides equally among all teams the overall
audience of the whole tournament. Formally,

\medskip \noindent \textbf{Uniform rule}, $U$: for each $A\in \mathbf{%
\mathcal{P}} $, and each $i\in N$, 
\begin{equation*}
U_{i}(A)=\frac{\left\vert \left\vert A\right\vert \right\vert }{n}.
\end{equation*}

The \textit{equal-split rule} divides the audience of each game equally,
among the two participating teams. Formally,

\medskip \noindent \textbf{Equal-split rule}, $ES$: for each $A \in \mathbf{%
\mathcal{P}}$, and each $i\in N$, 
\begin{equation*}
ES_{i}(A) =\frac{\alpha _{i}}{2}.
\end{equation*}

For each game, \textit{concede-and-divide} concedes each team its number of
hard-core fans and divides equally the rest. We introduce it in an
equivalent and simpler way.\footnote{%
See Berganti\~{n}os and Moreno-Ternero (2020a) for a detailed discussion.}
Formally,

\medskip \noindent \textbf{Concede-and-divide}, $CD$: for each $A\in \mathbf{%
\mathcal{P}}$, and each $i\in N$, 
\begin{equation*}
CD_{i}(A)=\frac{\left( n-1\right) \alpha _{i}-\left\vert \left\vert
A\right\vert \right\vert }{n-2}.
\end{equation*}


We now introduce a family of rules studied in Berganti\~{n}os and
Moreno-Ternero (2022a; 2022b). For each $\lambda \in \mathbb{R}$ and each
game $\left( i,j\right) ,$ $S^{\lambda }$ divides the audience $a_{ij}$
among the teams $i$ and $j$ proportionally to $\left( 1-\lambda ,\lambda
\right) $. Formally, for each $A\in \mathcal{P}$ and each $i\in N,$ 
\begin{equation*}
S_{i}^{\lambda }\left( A\right) =\sum_{j\in N\backslash \left\{ i\right\}
}\left( 1-\lambda \right) a_{ij}+\sum_{j\in N\backslash \left\{ i\right\}
}\lambda a_{ji}.
\end{equation*}

The \textit{equal-split rule} corresponds to the case where $\lambda =0.5.$
When $\lambda =0$ all the audience is assigned to the home team and when $%
\lambda =1$ all the audience is assigned to the away team. We say that $R$
is a \textbf{split rule} if $R\in \left\{ S^{\lambda }:\lambda \in \left[ 0,1%
\right] \right\} .$ We say that $R$ is a \textbf{generalized split rule} if $%
R\in \left\{ S^{\lambda }:\lambda \in \mathbb{R}\right\} .$

\bigskip

The next family of rules, studied in Berganti\~{n}os and Moreno-Ternero
(2022c; 2023a), contains all the rules defined above. The amount received by
each team $i$ has three parts: one depending on the home audience of this
team $\left( \sum\limits_{j\in N\backslash \left\{ i\right\} }a_{ij}\right)
, $ other depending on the away audience of this team $\left(
\sum\limits_{j\in N\backslash \left\{ i\right\} }a_{ji}\right) ,$ and the
third part depending on the total audience of the tournament $\left(
\left\vert \left\vert A\right\vert \right\vert \right) .$ Formally,

\medskip \noindent \textbf{General rules }$\left\{ G^{xyz}\right\}
_{x+y+nz=1}$. For each trio $x,y,z\in \mathbb{R}$ with $x+y+nz=1,$ each $%
A\in \mathbf{\mathcal{P}}$, and each $i\in N$, 
\begin{equation*}
G_{i}^{xyz}(A)=x\sum_{j\in N\backslash \left\{ i\right\} }a_{ij}+y\sum_{j\in
N\backslash \left\{ i\right\} }a_{ji}+z\left\vert \left\vert A\right\vert
\right\vert .
\end{equation*}

The next table yields the value for the trio $x,y,z$, so that we obtain the
three focal rules introduced above. 
\begin{equation*}
\begin{tabular}{cccc}
& $x$ & $y$ & $z$ \\ 
\textit{Uniform rule} & $0$ & $0$ & $\frac{1}{n}$ \\ 
\textit{Equal-split rule} & $\frac{1}{2}$ & $\frac{1}{2}$ & $0$ \\ 
\textit{Concede-and-divide} & $\frac{n-1}{n-2}$ & $\frac{n-1}{n-2}$ & $\frac{%
-1}{n-2}$%
\end{tabular}%
\end{equation*}

The next three families of rules, which have been studied in Berganti\~{n}os
and Moreno-Ternero (2021; 2022a; 2022b) are also included in the family of 
\textit{general rules}, as they are defined via convex combinations of two
of the focal rules.

\textbf{EC rules} $\left\{ EC^{\lambda }\right\} _{\lambda \in \left[ 0,1%
\right] }$. For each $\lambda \in \left[ 0,1\right] ,$ each $A\in \mathbf{%
\mathcal{P}}$, and each $i\in N$, 
\begin{equation*}
EC_{i}^{\lambda }\left( A\right) =\lambda ES_{i}\left( A\right) +(1-\lambda
)CD_{i}\left( A\right) .
\end{equation*}


\textbf{UC rules} $\left\{ UC^{\lambda }\right\} _{\lambda \in \left[ 0,1%
\right] }$. For each $\lambda \in \left[ 0,1\right] ,$ each $A\in \mathbf{%
\mathcal{P}}$, and each $i\in N$, 
\begin{equation*}
UC_{i}^{\lambda }\left( A\right) =\lambda U_{i}\left( A\right) +(1-\lambda
)CD_{i}\left( A\right) .
\end{equation*}

\textbf{UE rules} $\left\{ UE^{\lambda }\right\} _{\lambda \in \left[ 0,1%
\right] }$. For each $\lambda \in \left[ 0,1\right] ,$ each $A\in \mathbf{%
\mathcal{P}}$, and each $i\in N$, 
\begin{equation*}
UE_{i}^{\lambda }\left( A\right) =\lambda U_{i}\left( A\right) +(1-\lambda
)ES_{i}\left( A\right) .
\end{equation*}


Berganti\~{n}os and Moreno-Ternero (2020a) prove that, for each $A\in 
\mathcal{P}$, 
\begin{equation*}
ES(A)=\frac{n}{2\left( n-1\right) }U(A)+\frac{n-2}{2\left( n-1\right) }CD(A).
\end{equation*}

Thus, 
\begin{eqnarray*}
\left\{ UC^{\lambda }\right\} _{\lambda \in \left[ 0,1\right] } &=&\left\{
UE^{\lambda }\right\} _{\lambda \in \left[ 0,1\right] }\cup \left\{
EC^{\lambda }\right\} _{\lambda \in \left[ 0,1\right] },\text{ and } \\
ES &=&\left\{ UE^{\lambda }\right\} _{\lambda \in \left[ 0,1\right] }\cap
\left\{ EC^{\lambda }\right\} _{\lambda \in \left[ 0,1\right] }.
\end{eqnarray*}

\bigskip

We now give an example to illustrate the various allocation rules.

\bigskip

\begin{example}
\label{example three teams} Let $A\in \mathbf{\mathcal{P}}$ be such that 
\begin{equation*}
A=\left( 
\begin{array}{ccc}
0 & 1200 & 1030 \\ 
750 & 0 & 140 \\ 
630 & 210 & 0%
\end{array}%
\right)
\end{equation*}

Then, $||A||=3960$ and $\left( \alpha _{i}\right) _{i\in N}=\left(
3610,2300,2010\right) .$

The three focal rules yield the following allocations in this example:%
\begin{equation*}
\begin{tabular}{cccc}
Rule & Team 1 & Team 2 & Team 3 \\ 
$U$ & 1320 & 1320 & 1320 \\ 
$ES$ & 1805 & 1150 & 1005 \\ 
$CD$ & 3260 & 640 & 60%
\end{tabular}%
\end{equation*}

We also compute the allocations from several generalized split rules $%
S^{\lambda }$:%
\begin{equation*}
\begin{tabular}{cccc}
$\lambda $ & Team 1 & Team 2 & Team 3 \\ 
0 & 2230 & 890 & 840 \\ 
0.2 & 2060 & 994 & 906 \\ 
1 & 1380 & 1410 & 1170 \\ 
4 & -1170 & 2970 & 2160%
\end{tabular}%
\end{equation*}

And the allocations from several general rules $G^{xyz}$: 
\begin{equation*}
\begin{tabular}{cccc}
$\left( x,y,z\right) $ & Team 1 & Team 2 & Team 3 \\ 
$\left( 0.5,0.2,0.1\right) $ & 1787 & 1123 & 1050 \\ 
$\left( 1,3,-1\right) $ & 2410 & 1160 & 390%
\end{tabular}%
\end{equation*}

Finally, we compute the allocations from $EC^{\lambda },$ $UC^{\lambda },$
and $UE^{\lambda }$ when $\lambda =0.5$: 
\begin{equation*}
\begin{tabular}{cccc}
Rule & Team 1 & Team 2 & Team 3 \\ 
$EC^{0.5}$ & 2532.5 & 895 & 532.5 \\ 
$UC^{0.5}$ & 2290 & 980 & 690 \\ 
$UE^{0.5}$ & 1562.5 & 1235 & 1162.5%
\end{tabular}%
\end{equation*}
\end{example}



\subsection{Axioms}

We now introduce the axioms considered in this paper.


Our key axiom is \textit{anonymity}, which says that a permutation of the
set of teams equally permutes the allocation. Formally, let $\sigma $ be a
permutation of the set of teams. Thus, $\sigma :N\rightarrow N$ such that $%
\sigma \left( i\right) \neq \sigma \left( j\right) $ when $i\neq j.$ Given a
permutation $\sigma $ and $A\in \mathcal{P}$, we define the problem $%
A^{\sigma }$ where for each $i,j\in N,$ $a_{ij}^{\sigma }=a_{\sigma \left(
i\right) \sigma \left( j\right) }.$

\textbf{Anonymity }$\left( AN\right) $: For each $A\in \mathcal{P}$, each
permutation $\sigma ,$ and each $i\in N$, 
\begin{equation*}
R_{i}(A)=R_{\sigma \left( i\right) }(A^{\sigma }).
\end{equation*}

The next axiom says that revenues should be additive on $A$. Formally, 

\textbf{Additivity }$\left( AD\right) $: For each pair $A$ and $A^{\prime
}\in \mathcal{P}$, 
\begin{equation*}
R\left( A+A^{\prime }\right) =R(A)+R\left( A^{\prime }\right) .
\end{equation*}


All the results presented in this paper involve the two axioms introduced
above. We also consider other axioms. The first three axioms represent the
three different ways of applying the principle of impartiality we mentioned
at the Introduction.

Equal treatment of equals says that if two teams have the same audiences,
when facing each of the other teams, then they should receive the same
amount. This axiom has been used in Berganti\~{n}os and Moreno-Ternero
(2020a; 2020b; 2021; 2022a; 2022b).

\textbf{Equal treatment of equals }$\left( ETE\right) $: For each $A\in 
\mathcal{P}$, and each pair $i,j\in N$ such that $a_{ik}=a_{jk}$, and $%
a_{ki}=a_{kj}$, for each $k\in N\setminus \{i,j\}$, 
\begin{equation*}
R_{i}(A)=R_{j}(A).
\end{equation*}


In Berganti\~{n}os and Moreno-Ternero (2022a; 2022b; 2022c) we also require
something more to consider two teams as equals. Not only the two teams must
have the same audiences when facing each of the other teams, but they also
must have the same audiences when facing themselves at each stadium.
Formally,

\textbf{Weak equal treatment of equals }$\left( WETE\right) $: For each $%
A\in \mathcal{P}$, and each pair $i,j\in N$ such that $a_{ij}=a_{ji},$ $%
a_{ik}=a_{jk}$, and $a_{ki}=a_{kj}$, for each $k\in N\setminus \{i,j\}$, 
\begin{equation*}
R_{i}(A)=R_{j}(A).
\end{equation*}


Finally, \textit{symmetry} requires less to consider two teams as
equal/symmetric. More precisely, the axiom says that if two teams have the
same aggregate audience in the tournament, then they should receive the same
amount. This axiom has been considered in Berganti\~{n}os and Moreno-Ternero
(2021).

\textbf{Symmetry} $\left( SYM\right) $: For each $A\in \mathcal{P}$, and
each pair $i,j\in N$, such that $\alpha _{i}\left( A\right) =\alpha
_{j}\left( A\right) $,%
\begin{equation*}
R_{i}(A)=R_{j}(A).
\end{equation*}


The following axioms strengthen the previous ones by saying that if the
audience of team $i$ is, game by game, not smaller than the audience of team 
$j$, then that team $i$ should not receive less than team $j$.

\textbf{Order preservation }$\left( OP\right) $: For each $A\in \mathcal{P}$
and each pair $i,j\in N$, such that, for each $k\in N\backslash \left\{
i,j\right\} $, $a_{ik}\geq a_{jk}$ and $a_{ki}\geq a_{kj}$ we have that 
\begin{equation*}
R_{i}(A)\geq R_{j}(A).
\end{equation*}

\textbf{Home order preservation }$\left( HOP\right) $: For each $A\in 
\mathcal{P}$ and each pair $i,j\in N$, such that, for each $k\in N\backslash
\left\{ i,j\right\} $, $a_{ik}\geq a_{jk}$, $a_{ki}\geq a_{kj}$, and $%
a_{ij}\geq a_{ji}$ we have that 
\begin{equation*}
R_{i}(A)\geq R_{j}(A).
\end{equation*}

\textbf{Away order preservation }$\left( AOP\right) $: For each $A\in 
\mathcal{P}$ and each pair $i,j\in N$, such that, for each $k\in N\backslash
\left\{ i,j\right\} $, $a_{ik}\geq a_{jk}$, $a_{ki}\geq a_{kj}$, and $%
a_{ji}\geq a_{ij}$ we have that 
\begin{equation*}
R_{i}(A)\geq R_{j}(A).
\end{equation*}


The next two axioms refer to the performance of the rule with respect to
somewhat pathological teams. First, \textit{null team} says that if a team
has a null audience, then such a team gets no revenue. Second, \textit{%
essential team} says that if only the games played by some team have
positive audience, then such a team should receive all its audience.
Formally,

\textbf{Null team }$\left( NT\right) $: For each $A\in \mathcal{P}$, and
each $i\in N$, such that $a_{ij}=0=a_{ji}$, for each $j\in N$, 
\begin{equation*}
R_{i}(A)=0.
\end{equation*}

\textbf{Essential team }$\left( ET\right) $: For each $A\in \mathcal{P}$ and
each $i\in N$ such that $a_{jk}=0$ for each pair $\left\{ j,k\right\} \in
N\backslash \left\{ i\right\} $, 
\begin{equation*}
R_{i}(A)=\alpha _{i}\left( A\right) .
\end{equation*}


The next three axioms provide natural (lower/upper) bounds to the amount a
team could receive.

The first one says that each team should receive, at most, the total
audience of the games it played.

\textbf{Maximum aspirations }$\left( MA\right) $: For each $A\in \mathcal{P}$
and each $i\in N$, 
\begin{equation*}
R_{i}(A)\leq \alpha _{i}\left( A\right) .
\end{equation*}


The second one says that each team should receive, at most, the total
audience of all games in the tournament.

\textbf{Weak upper bound }$\left( WUB\right) $: For each $A\in \mathcal{P}$
and each $i\in N$, 
\begin{equation*}
R_{i}(A)\leq \left\vert \left\vert A\right\vert \right\vert .
\end{equation*}

The third axiom says that no team should receive negative awards.

\textbf{Non-negativity }$\left( NN\right) $. For each $A\in \mathcal{P}$ and
each $i\in N,$ 
\begin{equation*}
R_{i}(A)\geq 0.
\end{equation*}


The next result, whose straightforward proof we omit, summarizes the logical
relations between the axioms introduced above.\footnote{$A\rightarrow B$
means that if a rule satisfies property $A,$ than it also satisfies property 
$B.$}

\begin{proposition}
The following statements hold:

\label{rel axioms}$\left( 1\right) $ $AN\rightarrow WETE\leftarrow
ETE\leftarrow SYM.$

$\left( 2\right) $ $ETE\leftarrow OP\rightarrow WETE.$

$\left( 3\right) $ $MA\rightarrow WUB\leftarrow NN.$

$\left( 4\right) $ $OP\rightarrow \left\{ HOP,AOP\right\} \rightarrow OP$
\end{proposition}


As for the impartiality axioms, we note first that there exist rules
satisfying $WETE$ but violating $AN$. For instance, the rule that divides
the total audience equally among team 1 and all equally-deserving teams,
according to $WETE.$ Obviously this rule satisfies $WETE.$ The next example
shows that it does not satisfy $AN.$

\begin{example}
Let $A\in \mathbf{\mathcal{P}}$ be such that 
\begin{equation*}
A=\left( 
\begin{array}{ccc}
0 & 6 & 4 \\ 
6 & 0 & 4 \\ 
2 & 2 & 0%
\end{array}%
\right)
\end{equation*}

Then $R\left( A\right) =\left( 12,12,0\right) .$

Let $\sigma $ be such that $\sigma \left( 1\right) =2,$ $\sigma \left(
2\right) =3,$ and $\sigma \left( 3\right) =1.$ Now 
\begin{equation*}
A^{\sigma }=\left( 
\begin{array}{ccc}
0 & 4 & 6 \\ 
2 & 0 & 2 \\ 
6 & 4 & 0%
\end{array}%
\right)
\end{equation*}

Then $R\left( A^{\sigma }\right) =\left( 24,0,0\right) .$ As $R_{1}\left(
A\right) \neq R_{2}\left( A^{\sigma }\right) $ we deduce that $R$ does not
satisfy $AN.$
\end{example}

Likewise, there exist rules satisfying $ETE$ but violating $AN$. For
instance, the rule that divides the total audience equally among team 1 and
all equally-deserving teams, according to $ETE.$ Obviously this rule
satisfies $ETE.$ The previous example shows that it does not satisfy $AN$
(notice that teams 1 and 2 are equally-deserving teams, according to $ETE).$

Finally, there exist rules satisfying $AN$ but violating $ETE$. For
instance, the rule $S^{1}$ introduced above. By Theorem 1 in Berganti\~{n}os
and Moreno-Ternero (2023a), $S^{1}$ satisfies $AN.$ The next example shows
that $S^{1}$ does not satisfy $ETE,$

\begin{example}
Let $A\in \mathbf{\mathcal{P}}$ be such that 
\begin{equation*}
A=\left( 
\begin{array}{ccc}
0 & 4 & 6 \\ 
2 & 0 & 6 \\ 
3 & 3 & 0%
\end{array}%
\right)
\end{equation*}

Teams 1 and 2 satisfy the conditions of $ETE.$ Nevertheless, $S^{1}\left(
A\right) =\left( 5,7,12\right) .$
\end{example}


\section{New characterizations of well-known families of rules}


Our first result is a counterpart of Theorem 1 in Berganti\~{n}os and
Moreno-Ternero (2020a), obtained by replacing \textit{equal treatment of
equals} therein by \textit{anonymity}.

\begin{theorem}
\label{char gen MS}The following statements hold:

$\left( 1\right) $ A rule satisfies \textit{additivity}, \textit{anonymity},
and \textit{null team} if and only if it is a generalized split rule.

$\left( 2\right) $ A rule satisfies \textit{additivity}, \textit{anonymity},
and \textit{essential team} if and only if it is concede-and divide.
\end{theorem}

\begin{proof}
We first discuss the main ideas of this proof (which will also be considered
in the rest of the proofs). All our theorems characterize families of rules
(in some cases, a single rule) with a combination of several axioms that
include \textit{additivity} and \textit{anonymity}. The proof of the
theorems has two parts. First, to prove that each rule of the family
satisfies the combination of the axioms. This is made by checking that the
rule satisfies the formula stated by the axiom. Second, to prove that if a
rule satisfies the combination of axioms, then the rule belongs to the
family. As \textit{additivity} is always one of the axioms, the proof of
this part is made in two steps. In the first step, we decompose a general
problem in simple problems. We then characterize the rules satisfying the
combination of our axioms for simple problems. In the second one, we extend
the rules from simple problems to general problems using additivity. This
two-step procedure is quite standard in the literature of resource
allocation. The difference lies in the simple problems considered and the
axioms used in their characterizations.

Theorem 1 in Berganti\~{n}os and Moreno-Ternero (2022a) states that a rule
satisfies \textit{anonymity} and \textit{additivity} if and only if it is a 
\textit{general rule}. Thus, it only remains to show that a \textit{general
rule} satisfies \textit{null team} if and only if it is a \textit{%
generalized split rule}, and that \textit{concede-and divide} is the only 
\textit{general rule} that satisfies \textit{essential team}.

Let $R$ be a general rule, i.e., $R=G^{xyz}$ for some $x,y,z\in R$ with $%
x+y+nz=1$. For each pair $i,j\in N$, with $i\neq j$, let $\boldsymbol{1}%
^{ij}\in \mathbf{\mathcal{P}}$ denote the matrix with the following entries: 
\begin{equation*}
\boldsymbol{1}_{kl}^{ij}=\left\{ 
\begin{tabular}{cc}
$1$ & if $\left( k,l\right) =\left( i,j\right) $ \\ 
0 & otherwise.%
\end{tabular}%
\right.
\end{equation*}

The proof of Theorem 1 in Berganti\~{n}os and Moreno-Ternero (2022a) shows
that 
\begin{eqnarray}
z &=&R_{k}\left( \boldsymbol{1}^{ij}\right) \text{ with }k\in N\backslash
\left\{ i,j\right\}  \label{formula x y z} \\
x &=&R_{i}\left( \boldsymbol{1}^{ij}\right) -z\text{ and }  \notag \\
y &=&R_{j}\left( \boldsymbol{1}^{ij}\right) -z.  \notag
\end{eqnarray}

Besides, $z,$ $x$, and $y$ do not depend on $i,j,k.$


Now, by \textit{null team}, $z=0$ and hence $x=1-y.$ Then, 
\begin{equation*}
R_{i}(A)=G^{\left( 1-y\right) y0}=\left( 1-y\right) \sum_{j\in N\backslash
\left\{ i\right\} }a_{ij}+y\sum_{j\in N\backslash \left\{ i\right\}
}a_{ji}=S_{i}^{y}\left( A\right) .
\end{equation*}

Hence, $R$ is a \textit{generalized split rule}, which concludes the proof
of statement $\left( 1\right)$.


Now, by \textit{essential team}, $R_{i}\left( \boldsymbol{1}^{ij}\right)
=R_{j}\left( \boldsymbol{1}^{ij}\right) =1.$ Hence, $z=\frac{-1}{n-2}.$
Then, by \textit{additivity}, $R\left( A\right) =CD\left( A\right) $, which
concludes the proof of statement $\left( 2\right) $.
\end{proof}

\bigskip


\begin{remark}
\label{indep gen MS} The axioms used in Theorem \ref{char gen MS} are
independent.

$\left( 1\right) $ The \textit{uniform rule} satisfies $AD$ and $AN$ but
fails $NT.$

Let $R^{1}$ be the rule in which the audience of each game goes to the team
with the lowest number. Namely, for each $A\in \mathcal{P}$, and each $i\in
N,$ 
\begin{equation*}
R_{i}^{1}(N,A)=\sum_{j\in N:j>i}(a_{ij}+a_{ji}).
\end{equation*}

Then, $R^{1}$ satisfies $NT$ and $AD$ but fails $AN.$

Let $R^{2}$ be the rule in which the audience of each game is divided among
the teams playing the game proportionally to their audiences in the games
played against the other teams. Namely, for each $A\in \mathcal{P}$, and
each $i\in N,$%
\begin{equation*}
R_{i}^{2}(N,A)=\sum_{j\in N\backslash \left\{ i\right\} }\frac{%
\sum\limits_{k\in N\backslash \left\{ i,j\right\} }\left(
a_{ik}+a_{ki}\right) }{\sum\limits_{k\in N\backslash \left\{ i,j\right\}
}\left( a_{ik}+a_{ki}\right) +\sum\limits_{k\in N\backslash \left\{
i,j\right\} }\left( a_{jk}+a_{kj}\right) }\left[ a_{ij}+a_{ji}\right] .
\end{equation*}

Then, $R^{2}$ satisfies $AN$ and $NT$ but fails $AD.$ \bigskip

$\left( 2\right) $ The \textit{uniform rule} satisfies $AD$ and $AN$ but
fails $ET.$

Let $R^{3}$ be such that, for each pair $i,j\in N$, with $i\neq j$, each $%
\boldsymbol{1}^{ij}\in \mathbf{\mathcal{P}}$ and each $k\in N,$ 
\begin{equation*}
R_{k}^{3}\left( N,\boldsymbol{1}^{ij}\right) =\left\{ 
\begin{tabular}{ll}
$1$ & if $k\in \left\{ i,j\right\} $ \\ 
$-1$ & if $k=\min \left\{ l:l\in N\backslash \left\{ i,j\right\} \right\} $
\\ 
$0$ & otherwise%
\end{tabular}%
\right.
\end{equation*}%
We extend $R^{3}$ to all problems using \textit{additivity}. Namely, $%
R^{3}\left( A\right) =\sum\limits_{i,j\in N:i\neq j}a_{ij}R^{3}\left( 
\boldsymbol{1}^{ij}\right) .$

Then, $R^{3}$ satisfies $AD$ and $ET$ but fails $AN.$

Let $\mathbf{\mathcal{P}}^{\prime }$ be the set of problems having at least
one \textit{essential team}. Let $R^{4}$ be such that it coincides with 
\textit{concede-and-divide} on $\mathbf{\mathcal{P}}^{\prime }$ and with the 
\textit{uniform rule} on $\mathbf{\mathcal{P}}\backslash \mathbf{\mathcal{P}}%
^{\prime }.$

Then, $R^{4}$ satisfies $AN$ and $ET$ but fails $AD.$
\end{remark}


In many resource allocation problems, some relevant rules are also
characterized with the help of an impartiality axiom. Usually, if we replace
the impartiality axiom by \textit{anonymity} in the characterizations, the
result still holds.\footnote{%
For instance, the Shapley value is characterized by the combination of
efficiency, null player, additivity and either symmetry or anonymity.} We
stress that replacing \textit{equal treatment of equals} by \textit{anonymity%
} in Theorem 1 at Berganti\~{n}os and Moreno-Ternero (2020a) has different
impacts. With \textit{null team}, instead of characterizing the \textit{%
equal-split rule}, we characterize a whole family of rules containing it.
With \textit{essential team}, we simply obtain an alternative
characterization of \textit{concede-and-divide}. We also note that Theorem %
\ref{char gen MS} is the counterpart of Theorem 3.1 at Berganti\~{n}os and
Moreno-Ternero (2022a), replacing \textit{weak equal treatment of equals} by 
\textit{anonymity}. Finally, if we add one of the bound axioms introduced
above to the first statement of Theorem \ref{char gen MS}, we obtain a
characterization of \textit{split rules} instead of \textit{generalized
split rules}. 

\begin{corollary}
\label{char split rules}A rule satisfies \textit{additivity}, \textit{%
anonymity}, \textit{null team} and either \textit{\textit{maximum aspirations%
}}, \textit{weak upper bound} or \textit{non-negativity} if and only if it
is a \textit{split rule}.
\end{corollary}


\begin{proof}
Let $x,$ $y$, and $z$ be defined as in the proof of Theorem \ref{char gen MS}%
.

Let $i,j\in N$, with $i\neq j.$ By \textit{\textit{maximum aspirations}} or 
\textit{weak upper bound}, $x+z=R_{i}\left( \boldsymbol{1}^{ij}\right) \leq
1 $ and $y+z=R_{j}\left( \boldsymbol{1}^{ij}\right) \leq 1.$ As $z=0,$ we
deduce that $y\in \left[ 0,1\right] $ and $R=S^{y}.$

By \textit{non-negativity}, $x+z=R_{i}\left( \boldsymbol{1}^{ij}\right) \geq
0$ and $y+z=R_{j}\left( \boldsymbol{1}^{ij}\right) \geq 0.$ As $z=0,$ we
deduce that $y\in \left[ 0,1\right] $ and $R=S^{y}.$
\end{proof}


\section{Characterizations of new families of rules}

In this section we characterize new families of rules combining \textit{%
anonymity} and \textit{additivity} with the other axioms described above.
All the families have two parts. The first part is obtained by applying a
rule to the original problem. The second part is obtained by applying 
\textit{concede-and divide} to an associated problem where some audiences of
the original problem are nullified. 

Formally, for each $A\in \mathcal{P}$ and each $i\in N$ let $A^{i0}$ be the
matrix obtained from $A$ by nullifying the audiences of all the games $i$
played home. Namely, 
\begin{equation*}
a_{jk}^{i0}=\left\{ 
\begin{tabular}{ll}
$0$ & if $j=i$ \\ 
$a_{jk}$ & otherwise.%
\end{tabular}%
\right.
\end{equation*}

Besides, let $A^{0i}$ the matrix obtained from $A$ by nullifying the
audiences of all the games $i$ played away. Namely, 
\begin{equation*}
a_{jk}^{0i}=\left\{ 
\begin{tabular}{ll}
$0$ & if $k=i$ \\ 
$a_{jk}$ & otherwise.%
\end{tabular}%
\right.
\end{equation*}

\bigskip

We say that $R$ is an \textbf{extended }$EC$\textbf{\ rule} if there exist $%
x^{\prime },y^{\prime }\in \left[ 0,1\right] $ with $x^{\prime }+y^{\prime
}\geq 1$ such that, for each $i\in N$, 
\begin{equation*}
R_{i}\left( A\right) =\left\{ 
\begin{tabular}{ll}
$EC_{i}^{\lambda }\left( A\right) -\left\vert x^{\prime }-y^{\prime
}\right\vert CD_{i}\left( A^{i0}\right) $ & if $x^{\prime }\geq y^{\prime }$
\\ 
$EC_{i}^{\lambda }\left( A\right) -\left\vert x^{\prime }-y^{\prime
}\right\vert CD_{i}\left( A^{0i}\right) $ & if $x^{\prime }\leq y^{\prime }$%
\end{tabular}%
\right.
\end{equation*}%
where $\lambda =2-2\max \left\{ x^{\prime },y^{\prime }\right\} \in \left[
0,1\right] .$

We denote by $EC^{x^{\prime }y^{\prime }}$ the \textit{extended $EC$ rule}
associated to $x^{\prime }$ and $y^{\prime }$ as above. Notice that \textit{$%
EC$ rules} are \textit{extended $EC$ rules} (just take $x^{\prime
}=y^{\prime }).$

\bigskip

In the next theorem, we characterize the \textit{extended $EC$ rules}.


\begin{theorem}
\label{char AD+AN+MA}A rule satisfies \textit{additivity}, \textit{anonymity}
and maximum aspirations if and only if it is an extended $EC$ rule.
\end{theorem}

\begin{proof}
We have argued above that $ES$ is a \textit{generalized split rule}. Thus,
by Theorem \ref{char gen MS}, both $ES$ and $CD$ satisfy \textit{additivity}
and \textit{anonymity}. Therefore, so do all \textit{extended $EC$ rules}.

Let $EC^{x^{\prime }y^{\prime }}$ be an \textit{extended $EC$ rule}. We
first show that $EC^{x^{\prime }y^{\prime }}$ satisfies \textit{\textit{%
maximum aspirations}} for each matrix $\boldsymbol{1}^{ij}\in \mathbf{%
\mathcal{P}}$ with $i,j\in N$ such that $i\neq j.$ 

We consider several cases.

\begin{enumerate}
\item $x^{\prime }\geq y^{\prime }$. 
\begin{eqnarray*}
EC_{i}^{x^{\prime }y^{\prime }}\left( \boldsymbol{1}^{ij}\right) &=&\left(
2-2x^{\prime }\right) ES_{i}\left( \boldsymbol{1}^{ij}\right) +\left(
2x^{\prime }-1\right) CD_{i}\left( \boldsymbol{1}^{ij}\right) -\left\vert
x^{\prime }-y^{\prime }\right\vert CD_{i}\left( \left( \boldsymbol{1}%
^{ij}\right) ^{i0}\right) \\
&=&\left( 2-2x^{\prime }\right) \frac{1}{2}+\left( 2x^{\prime }-1\right)
=x^{\prime }\leq 1.
\end{eqnarray*}

Similarly, $EC_{j}^{x^{\prime }y^{\prime }}\left( \boldsymbol{1}^{ij}\right)
=x^{\prime }\leq 1.$

For each $k\in N\backslash \left\{ i,j\right\}$, 
\begin{equation*}
EC_{k}^{x^{\prime }y^{\prime }}\left( \boldsymbol{1}^{ij}\right) =\left(
2x^{\prime }-1\right) \frac{-1}{n-2}=\frac{1-2x^{\prime }}{n-2}\leq 0.
\end{equation*}

\item $x^{\prime }\leq y^{\prime }.$ 
\begin{equation*}
EC_{i}^{x^{\prime }y^{\prime }}\left( \boldsymbol{1}^{ij}\right) =\left(
2-2y^{\prime }\right) \frac{1}{2}+\left( 2y^{\prime }-1\right) -\left(
y^{\prime }-x^{\prime }\right) =x^{\prime }\leq 1.
\end{equation*}

Similarly, $EC_{j}^{x^{\prime }y^{\prime }}\left( \boldsymbol{1}^{ij}\right)
=x^{\prime }\leq 1.$

For each $k\in N\backslash \left\{ i,j\right\}$, 
\begin{equation*}
EC_{k}^{x^{\prime }y^{\prime }}\left( \boldsymbol{1}^{ij}\right) =\left(
2y^{\prime }-1\right) \frac{-1}{n-2}-\left( y^{\prime }-x^{\prime }\right) 
\frac{-1}{n-2}=\frac{1-x^{\prime }-y^{\prime }}{n-2}\leq 0.
\end{equation*}
\end{enumerate}

Then, $EC^{x^{\prime }y^{\prime }}$ satisfies \textit{maximum aspirations}
for $\boldsymbol{1}^{ij}\in \mathbf{\mathcal{P}}$. We now prove it in the
general case. Let $A\in \mathcal{P}$ and $i\in N.$ As $EC^{x^{\prime
}y^{\prime }}$ satisfies \textit{additivity}, 
\begin{eqnarray*}
EC_{i}^{x^{\prime }y^{\prime }}\left( A\right) &=&\sum_{j,k\in
N}a_{jk}EC_{i}^{x^{\prime }y^{\prime }}\left( \boldsymbol{1}^{jk}\right) \\
&=&\sum_{j\in N\backslash \left\{ i\right\} }\left( a_{ij}EC_{i}^{x^{\prime
}y^{\prime }}\left( \boldsymbol{1}^{ij}\right) +a_{ji}EC_{i}^{x^{\prime
}y^{\prime }}\left( \boldsymbol{1}^{ji}\right) \right) +\sum_{j,k\in
N\backslash \left\{ i\right\} }a_{jk}EC_{i}^{x^{\prime }y^{\prime }}\left( 
\boldsymbol{1}^{jk}\right) \\
&\leq &\sum_{j\in N\backslash \left\{ i\right\} }\left( a_{ij}+a_{ji}\right)
+\sum_{j,k\in N\backslash \left\{ i\right\} }0 \\
&=&\alpha _{i}\left( A\right) .
\end{eqnarray*}

Conversely, let $R$ be a rule satisfying the axioms from the statement. Let $%
A\in \mathcal{P}$. Given $i,j\in N,$ with $i\neq j$, let $x^{\prime
}=R_{i}\left( \boldsymbol{1}^{ij}\right) $ and $y^{\prime }=R_{j}\left( 
\boldsymbol{1}^{ij}\right) $ and $z=R_{k}\left( \boldsymbol{1}^{ij}\right) $
for $k\in N\backslash \left\{ i,j\right\} .$

By \textit{maximum aspirations}, $x^{\prime }\leq \alpha_{i}\left( 
\boldsymbol{1}^{ij}\right)= 1,$ $y^{\prime }\leq \alpha_{j}\left( 
\boldsymbol{1}^{ij}\right)=1$ and $z\leq \alpha_{k}\left( \boldsymbol{1}%
^{ij}\right)=0.$ As $z=\frac{1-x^{\prime }-y^{\prime }}{n-2}$, we deduce
that $x^{\prime }+y^{\prime }\geq 1.$ Then, $x^{\prime }\geq 0$ and $%
y^{\prime }\geq 0.$

Let $i\in N.$ By \textit{additivity}, 
\begin{eqnarray*}
R_{i}(A) &=&\sum_{j\in N}a_{ij}R_{i}\left( \boldsymbol{1}^{ij}\right)
+\sum_{j\in N}a_{ji}R_{i}\left( 1^{ji}\right) +\sum_{j,k\in N\backslash
\left\{ i\right\} }a_{jk}R_{i}\left( 1^{jk}\right) \\
&=&x^{\prime }\sum_{j\in N}a_{ij}+y^{\prime }\sum_{j\in N}a_{ji}+\frac{%
1-x^{\prime }-y^{\prime }}{n-2}\sum_{j,k\in N\backslash \left\{ i\right\}
}a_{jk}.
\end{eqnarray*}

We consider two cases.

Case 1. $\max \left\{ x^{\prime },y^{\prime }\right\} =x^{\prime }.$

Then, 
\begin{equation*}
R_{i}(A)=x^{\prime }\alpha _{i}\left( A\right) +\frac{1-2x^{\prime }}{n-2}%
\left( \left\vert \left\vert A\right\vert \right\vert -\alpha _{i}\left(
A\right) \right) +\left( y^{\prime }-x^{\prime }\right) \sum_{j\in N}a_{ji}-%
\frac{y^{\prime }-x^{\prime }}{n-2}\left( \left\vert \left\vert A\right\vert
\right\vert -\alpha _{i}\left( A\right) \right) .
\end{equation*}

Thus, 
\begin{eqnarray*}
x^{\prime }\alpha _{i}\left( A\right) +\frac{1-2x^{\prime }}{n-2}\left(
\left\vert \left\vert A\right\vert \right\vert -\alpha _{i}\left( A\right)
\right) &=&x^{\prime }\alpha _{i}\left( A\right) +\left( 2x^{\prime
}-1\right) \frac{\left( n-1\right) \alpha _{i}\left( A\right) -\left\vert
\left\vert A\right\vert \right\vert }{n-2}-\left( 2x^{\prime }-1\right)
\alpha _{i}\left( A\right) \\
&=&\left( 2-2x^{\prime }\right) ES_{i}\left( A\right) +\left( 2x^{\prime
}-1\right) CD_{i}\left( A\right) .
\end{eqnarray*}

As $x^{\prime },y^{\prime }\in \left[ 0,1\right] ,$ $x^{\prime }+y^{\prime
}\geq 1$ and $x^{\prime }\geq y^{\prime }$ we have that $x^{\prime }\in %
\left[ \frac{1}{2},1\right] .$ Let $\lambda =2-2x^{\prime }.$ Then $\lambda
\in \left[ 0,1\right] $.

Besides, 
\begin{eqnarray*}
\left( y^{\prime }-x^{\prime }\right) \sum_{j\in N}a_{ji}-\frac{y^{\prime
}-x^{\prime }}{n-2}\left( \left\vert \left\vert A\right\vert \right\vert
-\alpha _{i}\left( A\right) \right) &=&\left( y^{\prime }-x^{\prime }\right) 
\frac{\left( n-2\right) \sum\limits_{j\in N}a_{ji}+\alpha _{i}\left(
A\right) -\left\vert \left\vert A\right\vert \right\vert }{n-2} \\
&=&\left( y^{\prime }-x^{\prime }\right) CD_{i}\left( A^{i0}\right) .
\end{eqnarray*}

Now, 
\begin{equation*}
R_{i}(A)=EC^{2-2\max \left\{ x^{\prime },y^{\prime }\right\} }\left(
A\right) -\left\vert x^{\prime }-y^{\prime }\right\vert CD_{i}\left(
A^{i0}\right).
\end{equation*}

Case 2. $\max \left\{ x^{\prime },y^{\prime }\right\} =y^{\prime }.$

Then, 
\begin{equation*}
R_{i}(A)=y^{\prime }\alpha _{i}\left( A\right) +\frac{1-2y^{\prime }}{n-2}%
\left( \left\vert \left\vert A\right\vert \right\vert -\alpha _{i}\left(
A\right) \right) +\left( x^{\prime }-y^{\prime }\right) \sum_{j\in N}a_{ij}-%
\frac{x^{\prime }-y^{\prime }}{n-2}\left( \left\vert \left\vert A\right\vert
\right\vert -\alpha _{i}\left( A\right) \right) .
\end{equation*}

Similarly to Case 1, we deduce that 
\begin{equation*}
y^{\prime }\alpha _{i}\left( A\right) +\frac{1-2y^{\prime }}{n-2}\left(
\left\vert \left\vert A\right\vert \right\vert -\alpha _{i}\left( A\right)
\right) =EC^{2-2\max \left\{ x^{\prime },y^{\prime }\right\} }\left(
A\right) .
\end{equation*}

Besides, 
\begin{eqnarray*}
\left( x^{\prime }-y^{\prime }\right) \sum_{j\in N}a_{ij}-\frac{x^{\prime
}-y^{\prime }}{n-2}\left( \left\vert \left\vert A\right\vert \right\vert
-\alpha _{i}\left( A\right) \right) &=&\left( x^{\prime }-y^{\prime }\right) 
\frac{\left( n-2\right) \sum\limits_{j\in N}a_{ij}+\alpha _{i}\left(
A\right) -\left\vert \left\vert A\right\vert \right\vert }{n-2} \\
&=&\left( x^{\prime }-y^{\prime }\right) CD_{i}\left( A^{0i}\right) .
\end{eqnarray*}

Thus, 
\begin{equation*}
R_{i}(A)=EC^{2-2\max \left\{ x^{\prime },y^{\prime }\right\} }\left(
A\right) -\left\vert x^{\prime }-y^{\prime }\right\vert CD_{i}\left(
A^{0i}\right) .
\end{equation*}
\end{proof}

\bigskip

Theorem \ref{char AD+AN+MA} highlights the central role of \textit{%
concede-and-divide} as a rule to solve broadcasting problems. Extended $EC$
rules have two parts: one depending on the family of $EC$ rules $\left(
EC_{i}^{\lambda }\left( A\right) \right) $ and another depending on \textit{%
concede-and-divide} ($CD_{i}\left( A^{i0}\right) $ and $CD_{i}\left(
A^{0i}\right) ).$


\begin{remark}
\label{indep AD+AN+MA} The axioms used in Theorem \ref{char AD+AN+MA} are
independent.

The \textit{uniform rule} satisfies $AD$ and $AN$ but fails $MA.$

$R^{1}$, defined as in Remark \ref{indep gen MS}, satisfies $AD$ and $MA$
but fails $AN.$

For each $A\in \mathcal{P}$, let 
\begin{equation*}
R^{5}\left( A\right) =\left\{ 
\begin{tabular}{ll}
$ES\left( A\right) $ & when $\left\vert \left\vert A\right\vert \right\vert
\leq 10$ \\ 
$CD\left( A\right) $ & otherwise.%
\end{tabular}%
\right.
\end{equation*}

Then, $R^{5}$ satisfies $AN$ and $MA$ but fails $AD.$
\end{remark}


The $EC$ rules are characterized by the combination of \textit{additivity}, 
\textit{maximum aspirations} and either \textit{symmetry} or \textit{equal
treatment of equals} (e.g., Berganti\~{n}os and Moreno-Ternero, 2021,
2022a). Theorem \ref{char AD+AN+MA} shows that if we replace \textit{symmetry%
} or \textit{equal treatment of equals} by \textit{anonymity}, more rules
arise.

\bigskip

The next corollary states the effect of adding the three \textit{order
preservation} axioms to the statement of Theorem \ref{char AD+AN+MA}.

\begin{corollary}
\label{char AD+AN+WOP+MA} The following statements hold:

$\left( 1\right) $ A rule satisfies \textit{additivity}, \textit{anonymity},
maximum aspirations and home order preservation if and only if it is an
extended $EC$ rule with $x^{\prime }\geq y^{\prime }.$

$\left( 2\right) $ A rule satisfies \textit{additivity}, \textit{anonymity},
maximum aspirations and away order preservation if and only if it is an
extended $EC$ rule with $x^{\prime }\leq y^{\prime }.$

$\left( 3\right) $ A rule satisfies \textit{additivity}, \textit{anonymity},
maximum aspirations and order preservation if and only if it is an $EC$ rule.
\end{corollary}

\begin{proof}
$\left( 1\right) $ Let $EC^{x^{\prime }y^{\prime }}$ be an extended $EC$
rule with $x^{\prime }\geq y^{\prime }.$ By Theorem \ref{char AD+AN+MA}, $%
EC^{x^{\prime }y^{\prime }}$ satisfies \textit{additivity}, \textit{anonymity%
} and \textit{maximum aspirations}. As $x^{\prime }\geq y^{\prime }$, and $%
\left\{EC^{\lambda }:\lambda \in \left[ 0,1\right] \right\} \subset
\left\{UC^{\lambda }:\lambda \in \left[ 0,1\right] \right\} $, we deduce
from Theorem \ref{char AD+AN+WOP+WUB} below that $EC^{x^{\prime }y^{\prime
}} $ satisfies home order preservation.

Let $R$ be a rule satisfying all the axioms from the statement. By the proof
of Theorem \ref{char AD+AN+MA}, $R=EC^{x^{\prime }y^{\prime }}$ where $%
x^{\prime }=R_{i}\left( \boldsymbol{1}^{ij}\right) $ and $y^{\prime
}=R_{j}\left( \boldsymbol{1}^{ij}\right) $ for each pair $i,j\in N$ with $%
i\neq j.$ By \textit{home order preservation} we deduce that $x^{\prime
}\geq y^{\prime }.$

$\left( 2\right) $ It is similar to $\left( 1\right) .$

$\left( 3\right) $ We have argued that each $EC$ rule satisfies \textit{%
additivity}, \textit{anonymity} and \textit{maximum aspirations}. It is
obvious that they also satisfy \textit{order preservation}.

Let $R$ be a rule satisfying all the axioms from the statement. By the proof
of Theorem \ref{char AD+AN+MA}, $R=EC^{x^{\prime }y^{\prime }}$ where $%
x^{\prime }=R_{i}\left( \boldsymbol{1}^{ij}\right) $ and $y^{\prime
}=R_{j}\left( \boldsymbol{1}^{ij}\right) $ for each pair $i,j\in N$ with $%
i\neq j$. By Proposition \ref{rel axioms}, \textit{order preservation}
implies \textit{equal treatment of equals}. Thus, $x^{\prime }=y^{\prime }$
and hence $R$ is an \textit{$EC$ rule}.
\end{proof}

\bigskip

We say that $R$ is an \textbf{extended }$UC$\textbf{\ rule} if there exist $%
x^{\prime },y^{\prime }\in \mathbb{R}$ such that, for each $i\in N$, 
\begin{equation*}
R_{i}\left( A\right) =\left\{ 
\begin{tabular}{ll}
$UC_{i}^{\lambda }\left( A\right) -\left\vert x^{\prime }-y^{\prime
}\right\vert CD_{i}\left( A^{i0}\right) $ & if $x^{\prime }\geq y^{\prime }$
\\ 
$UC_{i}^{\lambda }\left( A\right) -\left\vert x^{\prime }-y^{\prime
}\right\vert CD_{i}\left( A^{0i}\right) $ & if $x^{\prime }\leq y^{\prime }$%
\end{tabular}%
\right.
\end{equation*}%
where $\max \left\{ x^{\prime },y^{\prime }\right\} \in \left[ \frac{1}{n},1%
\right] ,$ $\min \left\{ x^{\prime },y^{\prime }\right\} \in \left[ \frac{%
1-\max \left\{ x^{\prime },y^{\prime }\right\} }{n-1},\max \left\{ x^{\prime
},y^{\prime }\right\} \right] ,$ and $\lambda =\frac{n\left( 1-\max \left\{
x^{\prime },y^{\prime }\right\} \right) }{n-1}\in \left[ 0,1\right] $.

We denote by $UC^{x^{\prime }y^{\prime }}$ the rule associated to the
numbers $x^{\prime }$ and $y^{\prime }$, as above. Notice that $UC$ rules
are all \textit{extended $UC$ rules} (just take $x^{\prime }=y^{\prime }).$

\bigskip

In the next theorem we provide a characterization of \textit{extended $UC$
rules}.


\begin{theorem}
\label{char AD+AN+WOP+WUB} The following statements hold:

$\left( 1\right) $ A rule satisfies \textit{additivity}, \textit{anonymity}, 
\textit{home order preservation} and \textit{weak upper bound} if and only
if it is an extended $UC$ rule with $x^{\prime }\geq y^{\prime }.$

$\left( 2\right) $ A rule satisfies \textit{additivity}, \textit{anonymity}, 
\textit{away order preservation} and \textit{weak upper bound} if and only
if it is an extended $UC$ rule with $x^{\prime }\leq y^{\prime }$.
\end{theorem}

\begin{proof}
$\left( 1\right) $ It is obvious that $U$ satisfies \textit{additivity} and 
\textit{anonymity}. By Theorem \ref{char AD+AN+MA}, $CD$ satisfies \textit{%
additivity} and \textit{anonymity}. As $U$ and $CD$ satisfy \textit{%
additivity} and \textit{anonymity}, so do all \textit{extended $UC$ rules}.

As for \textit{home order preservation}, we focus first on each matrix $%
\boldsymbol{1}^{jk}$ where $j,k\in N$ with $j\neq k$. Let $UC^{x^{\prime
}y^{\prime }}$ be an extended $UC$ rule with $x^{\prime }\geq y^{\prime }$.
We compute $UC_{i}^{x^{\prime }y^{\prime }}\left( \boldsymbol{1}^{jk}\right) 
$ for each $i\in N.$ We consider several cases.

\begin{enumerate}
\item $j=i.$ Then, 
\begin{eqnarray}
UC_{i}^{x^{\prime }y^{\prime }}\left( \mathbf{1}^{ik}\right) &=&\frac{%
n\left( 1-x^{\prime }\right) }{n-1}U_{i}\left( \mathbf{1}^{ik}\right) +\frac{%
nx^{\prime }-1}{n-1}CD_{i}\left( \mathbf{1}^{ik}\right) -\left( x^{\prime
}-y^{\prime }\right) CD_{i}\left( \left( \mathbf{1}^{ik}\right) ^{i0}\right)
\label{formula UC x} \\
&=&\frac{n\left( 1-x^{\prime }\right) }{n-1}\frac{1}{n}+\frac{nx^{\prime }-1%
}{n-1}-0=x^{\prime }.  \notag
\end{eqnarray}

\item $k=i.$ Then, 
\begin{equation}
UC_{i}^{x^{\prime }y^{\prime }}\left( \mathbf{1}^{ji}\right) =\frac{n\left(
1-x^{\prime }\right) }{n-1}\frac{1}{n}+\frac{nx^{\prime }-1}{n-1}-\left(
x^{\prime }-y^{\prime }\right) =y^{\prime }.  \label{formula UC y}
\end{equation}

\item $i\notin \left\{ j,k\right\} .$ Then, 
\begin{equation}
UC_{i}^{x^{\prime }y^{\prime }}\left( \mathbf{1}^{jk}\right) =\frac{n\left(
1-x^{\prime }\right) }{n-1}\frac{1}{n}+\frac{nx^{\prime }-1}{n-1}\frac{-1}{%
n-2}-\left( x^{\prime }-y^{\prime }\right) \frac{-1}{n-2}=1-x^{\prime
}-y^{\prime }.  \label{formula UC 1-x-y}
\end{equation}
\end{enumerate}

Let $R^{x^{\prime }y^{\prime }}$ be any extended $UC$ rule with $x^{\prime
}\geq y^{\prime }$, $A\in \mathcal{P},$ and $i$ and $j$ as in the definition
of \textit{home order preservation}. By \textit{additivity}, 
\begin{eqnarray*}
UC_{i}^{x^{\prime }y^{\prime }}\left( A\right) &=&a_{ij}UC_{i}^{x^{\prime
}y^{\prime }}\left( \boldsymbol{\mathbf{1}}^{ij}\right)
+a_{ji}UC_{i}^{x^{\prime }y^{\prime }}\left( \mathbf{1}^{ji}\right)
+\sum_{k\in N\backslash \left\{ i,j\right\} }a_{ik}UC_{i}^{x^{\prime
}y^{\prime }}\left( \mathbf{1}^{ik}\right) \\
&&+\sum_{k\in N\backslash \left\{ i,j\right\} }a_{ki}UC_{i}^{x^{\prime
}y^{\prime }}\left( \mathbf{1}^{ki}\right) +\sum_{k,l\in N\backslash \left\{
i,j\right\} }a_{kl}UC_{i}^{x^{\prime }y^{\prime }}\left( \mathbf{1}%
^{kl}\right) \text{ and } \\
UC_{j}^{x^{\prime }y^{\prime }}\left( A\right) &=&a_{ij}UC_{j}^{x^{\prime
}y^{\prime }}\left( \boldsymbol{\mathbf{1}}^{ij}\right)
+a_{ji}UC_{j}^{x^{\prime }y^{\prime }}\left( \mathbf{1}^{ji}\right)
+\sum_{k\in N\backslash \left\{ i,j\right\} }a_{jk}UC_{j}^{x^{\prime
}y^{\prime }}\left( \mathbf{1}^{jk}\right) \\
&&+\sum_{k\in N\backslash \left\{ i,j\right\} }a_{kj}UC_{j}^{x^{\prime
}y^{\prime }}\left( \mathbf{1}^{kj}\right) +\sum_{k,l\in N\backslash \left\{
i,j\right\} }a_{kl}UC_{j}^{x^{\prime }y^{\prime }}\left( \mathbf{1}%
^{kl}\right) .
\end{eqnarray*}

By $\left( \ref{formula UC x}\right) $ and $\left( \ref{formula UC y}\right) 
$, 
\begin{eqnarray*}
a_{ij}UC_{i}^{x^{\prime }y^{\prime }}\left( \boldsymbol{\mathbf{1}}%
^{ij}\right) +a_{ji}UC_{i}^{x^{\prime }y^{\prime }}\left( \mathbf{1}%
^{ji}\right) &\geq &a_{ij}UC_{j}^{x^{\prime }y^{\prime }}\left( \boldsymbol{%
\mathbf{1}}^{ij}\right) +a_{ji}UC_{j}^{x^{\prime }y^{\prime }}\left( \mathbf{%
1}^{ji}\right) \Leftrightarrow \\
a_{ij}x^{\prime }+a_{ji}y^{\prime } &\geq &a_{ij}y^{\prime }+a_{ji}x^{\prime
}\Leftrightarrow \\
x^{\prime }\left( a_{ij}-a_{ji}\right) &\geq &y^{\prime }\left(
a_{ij}-a_{ji}\right) ,
\end{eqnarray*}%
which holds because $x^{\prime }\geq y^{\prime }$ and $a_{ij}\geq a_{ji}.$

By $\left( \ref{formula UC x}\right) ,$ for each $k\in N\backslash \left\{
i,j\right\} ,$ $UC_{i}^{x^{\prime }y^{\prime }}\left( \mathbf{1}^{ik}\right)
=UC_{j}^{x^{\prime }y^{\prime }}\left( \mathbf{1}^{jk}\right) =x^{\prime }.$
As $a_{ik}\geq a_{jk}$, for each $k\in N\backslash \left\{ i,j\right\} $, 
\begin{equation*}
\sum_{k\in N\backslash \left\{ i,j\right\} }a_{ik}UC_{i}^{x^{\prime
}y^{\prime }}\left( \mathbf{1}^{ik}\right) \geq \sum_{k\in N\backslash
\left\{ i,j\right\} }a_{jk}UC_{j}^{x^{\prime }y^{\prime }}\left( \mathbf{1}%
^{jk}\right) .
\end{equation*}

By $\left( \ref{formula UC y}\right) ,$ for each $k\in N\backslash \left\{
i,j\right\} ,$ $UC_{i}^{x^{\prime }y^{\prime }}\left( \mathbf{1}^{ki}\right)
=UC_{j}^{x^{\prime }y^{\prime }}\left( \mathbf{1}^{kj}\right) =y^{\prime }.$
As for each $k\in N\backslash \left\{ i,j\right\} $ $a_{ki}\geq a_{kj}$, 
\begin{equation*}
\sum_{k\in N\backslash \left\{ i,j\right\} }a_{ki}UC_{i}^{x^{\prime
}y^{\prime }}\left( \mathbf{1}^{ki}\right) \geq \sum_{k\in N\backslash
\left\{ i,j\right\} }a_{kj}UC_{j}^{x^{\prime }y^{\prime }}\left( \mathbf{1}%
^{kj}\right) .
\end{equation*}

By $\left( \ref{formula UC 1-x-y}\right) ,$ for each $k,l\in N\backslash
\left\{ i,j\right\} ,$ $UC_{i}^{x^{\prime }y^{\prime }}\left( \mathbf{1}%
^{kl}\right) =UC_{j}^{x^{\prime }y^{\prime }}\left( \mathbf{1}^{kl}\right)
=1-x^{\prime }-y^{\prime }.$ Then, 
\begin{equation*}
\sum_{k,l\in N\backslash \left\{ i,j\right\} }a_{kl}UC_{i}^{x^{\prime
}y^{\prime }}\left( \mathbf{1}^{kl}\right) =\sum_{k,l\in N\backslash \left\{
i,j\right\} }a_{kl}UC_{j}^{x^{\prime }y^{\prime }}\left( \mathbf{1}%
^{kl}\right) .
\end{equation*}

Then, $UC^{x^{\prime }y^{\prime }}$ satisfies \textit{home order preservation%
}.

Finally, as for \textit{weak upper bound}, and using arguments similar to
those used in the proof of Theorem \ref{char AD+AN+MA} for \textit{maximum
aspirations}, it suffices to show it for each matrix $\boldsymbol{1}^{ij}\in 
\mathbf{\mathcal{P}}$, with $i,j\in N$, such that $i\neq j$. 
We consider three cases.

\begin{enumerate}
\item By $\left( \ref{formula UC x}\right) ,$ $UC_{i}^{x^{\prime }y^{\prime
}}\left( \boldsymbol{1}^{ij}\right) =x^{\prime }\leq 1.$

\item By $\left( \ref{formula UC y}\right) ,$ $UC_{j}^{x^{\prime }y^{\prime
}}\left( \boldsymbol{1}^{ij}\right) =y^{\prime }\leq x^{\prime }\leq 1.$

\item By $\left( \ref{formula UC 1-x-y}\right) ,$ for each $k\in N\backslash
\left\{ i,j\right\} ,$ $UC_{k}^{x^{\prime }y^{\prime }}\left( \boldsymbol{1}%
^{ij}\right) =1-x^{\prime }-y^{\prime }.$ Now 
\begin{equation*}
1-x^{\prime }-y^{\prime }\leq 1\Leftrightarrow x^{\prime }+y^{\prime }\geq 0
\end{equation*}%
which holds always because $x^{\prime }\in \left[ \frac{1}{n},1\right] $ and 
$y^{\prime }\in \left[ \frac{1-x^{\prime }}{n-1},x^{\prime }\right] .$
\end{enumerate}

Then, $UC^{x^{\prime }y^{\prime }}$ satisfies \textit{weak upper bound}.
\bigskip

Conversely, let $R$ be a rule satisfying all the axioms from the statement.
Similarly to the proof of Theorem \ref{char AD+AN+MA}, we can find $%
x^{\prime },y^{\prime }$ such that, for each pair $i,j\in N$ with $i\neq j,$ 
$x^{\prime }=R_{i}\left( \boldsymbol{1}^{ij}\right) ,$ $y^{\prime
}=R_{j}\left( \boldsymbol{1}^{ij}\right) ,$ and for each $k\in N\backslash
\left\{ i,j\right\} ,$ $\frac{1-x^{\prime }-y^{\prime }}{n-2}=R_{k}\left( 
\boldsymbol{1}^{ij}\right) $.

By \textit{home order preservation}, 
\begin{equation*}
\frac{1-x^{\prime }-y^{\prime }}{n-2}\leq y^{\prime }\leq x^{\prime }.
\end{equation*}
Now,

\begin{itemize}
\item $x^{\prime }\geq \frac{1-x^{\prime }-y^{\prime }}{n-2}\geq \frac{%
1-2x^{\prime }}{n-2}.$ Notice that $x^{\prime }=y^{\prime }=\frac{1}{n}$
satisfies the previous conditions. Then, $x^{\prime }\geq \frac{1}{n}.$ By 
\textit{weak upper bound}, $x^{\prime }\leq 1.$ Thus, $\frac{1}{n}\leq
x^{\prime }\leq 1.$ Hence, $\max \left\{ x^{\prime },y^{\prime }\right\} \in %
\left[ \frac{1}{n},1\right] $ holds.

\item As $\frac{1-x^{\prime }-y^{\prime }}{n-2}\leq y^{\prime },$ it follows
that $\frac{1-x^{\prime }}{n-1}\leq y^{\prime }.$ Hence, $\min \left\{
x^{\prime },y^{\prime }\right\} \in \left[ \frac{1-\max \left\{ x^{\prime
},y^{\prime }\right\} }{n-1},\max \left\{ x^{\prime },y^{\prime }\right\} %
\right] $ holds.
\end{itemize}

Similarly to the proof of part $\left( 1\right) $ of Theorem \ref{char
AD+AN+MA}, we can prove that for each $i\in N,$ 
\begin{equation*}
R_{i}(A)=x^{\prime }\alpha _{i}\left( A\right) +\frac{1-2x^{\prime }}{n-2}%
\left( \left\vert \left\vert A\right\vert \right\vert -\alpha _{i}\left(
A\right) \right) -\left( x^{\prime }-y^{\prime }\right) CD_{i}\left(
A^{i0}\right) .
\end{equation*}

Now, 
\begin{eqnarray*}
x^{\prime }\alpha _{i}\left( A\right) +\frac{1-2x^{\prime }}{n-2}\left(
\left\vert \left\vert A\right\vert \right\vert -\alpha _{i}\left( A\right)
\right) &=&\frac{\left( n-2\right) x^{\prime }\alpha _{i}\left( A\right)
+\left( 1-2x^{\prime }\right) \left( \left\vert \left\vert A\right\vert
\right\vert -\alpha _{i}\left( A\right) \right) }{n-2} \\
&=&\frac{\left( 1-2x^{\prime }\right) \left\vert \left\vert A\right\vert
\right\vert }{n-2}+\frac{\left( nx^{\prime }-1\right) \alpha _{i}\left(
A\right) }{n-2} \\
&=&\frac{n\left( 1-x^{\prime }\right) }{n-1}\frac{\left\vert \left\vert
A\right\vert \right\vert }{n}+\left( \frac{1-nx^{\prime }}{\left( n-1\right)
\left( n-2\right) }\right) \left\vert \left\vert A\right\vert \right\vert +%
\frac{\left( nx^{\prime }-1\right) \alpha _{i}\left( A\right) }{n-2} \\
&=&\frac{n\left( 1-x^{\prime }\right) }{n-1}\frac{\left\vert \left\vert
A\right\vert \right\vert }{n}+\frac{nx^{\prime }-1}{n-1}\left( \frac{\left(
n-1\right) \alpha _{i}\left( A\right) -\left\vert \left\vert A\right\vert
\right\vert }{n-2}\right) .
\end{eqnarray*}

Let 
\begin{equation*}
\lambda =\frac{n\left( 1-x^{\prime }\right) }{n-1}.
\end{equation*}

As $\frac{1}{n}\leq x^{\prime }\leq 1$, it follows that $0\leq \lambda \leq
1.$ Thus, 
\begin{equation*}
x^{\prime }\alpha _{i}\left( A\right) +\frac{1-2x^{\prime }}{n-2}\left(
\left\vert \left\vert A\right\vert \right\vert -\alpha _{i}\left( A\right)
\right) =UC_{i}^{\lambda }\left( A\right) .
\end{equation*}

$\left( 2\right) $ Using arguments similar to those used in $\left( 1\right) 
$ we can prove that all extended $UC$ rules with $x^{\prime }\le y^{\prime }$
satisfy \textit{additivity}, \textit{anonymity}, \textit{weak upper bound}
and \textit{away order preservation}. Conversely, let $R$ be a rule
satisfying all those axioms. Similarly to the proof of Theorem \ref{char
AD+AN+MA}, we can find $x^{\prime },y^{\prime }$ such that, for each pair $%
i,j\in N$ with $i\neq j,$ $x^{\prime }=R_{i}\left( \boldsymbol{1}%
^{ij}\right) ,$ $y^{\prime }=R_{j}\left( \boldsymbol{1}^{ij}\right) ,$ and
for each $k\in N\backslash \left\{ i,j\right\} ,$ $\frac{1-x^{\prime
}-y^{\prime }}{n-2}=R_{k}\left( \boldsymbol{1}^{ij}\right) $.

Similarly to $\left( 1\right) $ above, we can prove that $y^{\prime }\in %
\left[ \frac{1}{n},1\right] $ and $x^{\prime }\in \left[ \frac{1-y^{\prime }%
}{n-1},y^{\prime }\right] $.

Similarly to the proof of part $\left( 2\right) $ of Theorem \ref{char
AD+AN+MA}, we can prove that for each $i\in N,$ 
\begin{equation*}
R_{i}(A)=y^{\prime }\alpha _{i}\left( A\right) +\frac{1-2y^{\prime }}{n-2}%
\left( \left\vert \left\vert A\right\vert \right\vert -\alpha _{i}\left(
A\right) \right) -\left( y^{\prime }-x^{\prime }\right) CD_{i}\left(
A^{0i}\right) .
\end{equation*}

Similarly to Case 1, we can prove that 
\begin{equation*}
y^{\prime }\alpha _{i}\left( A\right) +\frac{1-2y^{\prime }}{n-2}\left(
\left\vert \left\vert A\right\vert \right\vert -\alpha _{i}\left( A\right)
\right) =UC_{i}^{\lambda }\left( A\right),
\end{equation*}
where 
\begin{equation*}
\lambda =\frac{n\left( 1-y^{\prime }\right) }{n-1},
\end{equation*}
as desired.
\end{proof}

\bigskip

Theorem \ref{char AD+AN+WOP+WUB} also highlights the central role of \textit{%
concede-and-divide} as a rule to solve broadcasting problems. Extended $UC$
rules have two parts: one depending on the family of $UC$ rules $\left(
UC_{i}^{\lambda }\left( A\right) \right) $ and another depending on \textit{%
concede-and-divide} ($CD_{i}\left( A^{i0}\right) $ and $CD_{i}\left(
A^{0i}\right) ).$

\begin{remark}
\label{indep AD+AN+WOP+WUB} The axioms used in Theorem \ref{char
AD+AN+WOP+WUB} are independent.

$\left( 1\right) $ 
For each pair $\left\{ i,j\right\} \in N$, with $i\neq j$, and each $k\in N$%
, let 
\begin{equation*}
R_{k}^{6}\left( \boldsymbol{1}^{ij}\right) =\left\{ 
\begin{tabular}{ll}
2 & if $k\in \left\{ i,j\right\} $ \\ 
$\frac{-3}{n-2}$ & otherwise.%
\end{tabular}%
\right.
\end{equation*}%
We extend $R^{6}$ to each problem $A$ using \textit{additivity}. Namely, $%
R^{6}\left( A\right) =\sum\limits_{\left\{ i,j\right\} \subset
N}a_{ij}R^{3}\left( \boldsymbol{1}^{ij}\right) .$

Then, $R^{6}$ satisfies $AD,$ $AN$ and $HOP$ \textit{\ but fails }$WUB$%
\textit{.}

For each pair $\left\{ i,j\right\} \in N$, with $i\neq j$, and each $k\in N$%
, let 
\begin{equation*}
R_{k}^{7}\left( \boldsymbol{1}^{ij}\right) =\left\{ 
\begin{tabular}{ll}
0 & if $k\in \left\{ i,j\right\} $ \\ 
$\frac{1}{n-2}$ & otherwise%
\end{tabular}%
\right.
\end{equation*}%
We extend $R^{7}$ to each problem $A$ using \textit{additivity}. Namely, $%
R^{7}\left( A\right) =\sum\limits_{\left\{ i,j\right\} \subset
N}a_{ij}R^{7}\left( \boldsymbol{1}^{ij}\right) .$

Then, $R^{7}$ satisfies $AD,$ $AN,$ and $WUB$ but fails $HOP$\textit{.}

Let $N=\left\{ 1,2,3\right\} $ and $R^{8}$ be the separable rule (see
Berganti\~{n}os and Moreno-Ternero, 2022c) where, for each $i,j\in N$, 
\begin{equation*}
\left( x^{ij}\right) _{i,j\in N}=\left( 
\begin{array}{ccc}
& 0.69 & 0.83 \\ 
0.63 &  & 0.90 \\ 
0.69 & 0.83 & 
\end{array}%
\right) \text{ and }\left( y^{ij}\right) _{i,j\in N}=\left( 
\begin{array}{ccc}
& 0.56 & 0.49 \\ 
0.49 &  & 0.35 \\ 
0.35 & 0.28 & 
\end{array}%
\right)
\end{equation*}

Then, $R^{8}$ satisfies $AD,$ $HOP$ and $WUB$ but fails $AN$\textit{.}

For each $A\in \mathcal{P}$, and each $i\in N,$ let 
\begin{equation*}
R_{i}^{9}(A)=\left\{ 
\begin{tabular}{ll}
$U\left( A\right) $ & if $\left\vert \left\vert A\right\vert \right\vert
\leq 10$ \\ 
$CD\left( A\right) $ & if $\left\vert \left\vert A\right\vert \right\vert
>10.$%
\end{tabular}%
\right.
\end{equation*}

Then, $R^{9}$ satisfies $AN,$ $HOP,$ and $WUB$ but fails $AD$. \bigskip

$\left( 2\right) $ $R^{6}$ satisfies $AD,$ $AN,$ $AOP$ \textit{\ but fails }$%
WUB$\textit{.}

$R^{7}$ satisfies $AD,$ $AN,$ and $WUB$ but fails $AOP$\textit{.}

Let $N=\left\{ 1,2,3\right\} $ and $R^{10}$ be the separable rule (see
Berganti\~{n}os and Moreno-Ternero, 2022c) where, for each $i,j\in N$, 
\begin{equation*}
\left( x^{ij}\right) _{i,j\in N}=\left( 
\begin{array}{ccc}
& 0.56 & 0.49 \\ 
0.49 &  & 0.35 \\ 
0.35 & 0.28 & 
\end{array}%
\right) \text{ and }\left( y^{ij}\right) _{i,j\in N}=\left( 
\begin{array}{ccc}
& 0.69 & 0.83 \\ 
0.63 &  & 0.90 \\ 
0.69 & 0.83 & 
\end{array}%
\right)
\end{equation*}

Then, $R^{10}$ satisfies $AD,$ $AOP$ and $WUB$ but fails $AN$\textit{.}

$R^{9}$ satisfies $AN,$ $AOP,$ and $WUB$ but fails $AD$.
\end{remark}

\bigskip

A trivial consequence of Theorem \ref{char AD+AN+WOP+WUB} and part 4 of
Proposition \ref{rel axioms} is that a rule satisfies \textit{additivity}, 
\textit{anonymity}, \textit{order preservation} and \textit{weak upper bound}
if and only if it is a \textit{$UC$ rule}. We do not stress this result
because \textit{anonymity} is actually redundant for this characterization,
as shown in Berganti\~{n}os and Moreno-Ternero (2022a).

\bigskip

We say that $R$ is an \textbf{extended }$UE$\textbf{\ rule} if there exist $%
x^{\prime },y^{\prime }\in \mathbb{R}$ such that for each $i\in N$ 
\begin{equation*}
R_{i}\left( A\right) =\left\{ 
\begin{tabular}{ll}
$UE_{i}^{\lambda }\left( A\right) -\left\vert x^{\prime }-y^{\prime
}\right\vert CD_{i}\left( A^{i0}\right) $ & if $x^{\prime }\geq y^{\prime }$
\\ 
$UE_{i}^{\lambda }\left( A\right) -\left\vert x^{\prime }-y^{\prime
}\right\vert CD_{i}\left( A^{0i}\right) $ & if $x^{\prime }\leq y^{\prime }$%
\end{tabular}%
\right.
\end{equation*}%
where $\max \left\{ x^{\prime },y^{\prime }\right\} \in \left[ \frac{1}{n},1%
\right] ,$ $\min \left\{ x^{\prime },y^{\prime }\right\} \in \left[ 0,1-\max
\left\{ x^{\prime },y^{\prime }\right\} \right] ,$ and $\lambda =\frac{%
n\left( 1-2\max \left\{ x^{\prime },y^{\prime }\right\} \right) }{n-1}\in %
\left[ \frac{-n}{n-2},1\right] $.

We denote by $UE^{x^{\prime }y^{\prime }}$ the rule associated to the
numbers $x^{\prime }$ and $y^{\prime }$, as above. Notice that all $UE$
rules are \textit{extended $UE$ rules} too (just take $x^{\prime }=y^{\prime
}$ and hence $x^{\prime }\in \left[ \frac{1}{n},0.5\right] $ and $\lambda
\in \left[ 0,1\right] ).$

\bigskip

The next theorem characterizes the \textit{extended $UE$ rules}.

\begin{theorem}
\label{char AD+AN+WOP+NN} The following statements hold:

$\left( 1\right) $ A rule satisfies \textit{additivity}, \textit{anonymity},
home \textit{order preservation} and \textit{non-negativity} if and only if
it is an extended $UE$ rule with $x^{\prime }\geq y^{\prime }.$

$\left( 2\right) $ A rule satisfies \textit{additivity}, \textit{anonymity},
away \textit{order preservation} and \textit{non-negativity} if and only if
it is an extended $UE$ rule with $x^{\prime }\leq y^{\prime }.$
\end{theorem}

\bigskip

\begin{proof}
$\left( 1\right) $ As $ES$ and $U$ satisfy \textit{additivity} and \textit{%
anonymity} (Theorem \ref{char AD+AN+MA} and Theorem \ref{char AD+AN+WOP+WUB}%
), so do all \textit{extended $UE$ rules}.

Using similar arguments to those used in the proof of part $\left( 1\right) $
of Theorem \ref{char AD+AN+WOP+WUB}, we can also prove that $UE^{x^{\prime
}y^{\prime }}$ with $x^{\prime }\geq y^{\prime }$ satisfies \textit{home
order preservation}. Finally, as for \textit{non-negativity}, and using
arguments similar to those used in the proof of Theorem \ref{char AD+AN+MA}
for \textit{maximum aspirations}, it suffices to prove that for each $j,k\in
N$ with $j\neq k$ and each $i\in N,$ $UE_{i}^{x^{\prime }y^{\prime }}\left( 
\mathbf{1}^{jk}\right) \geq 0.$ We consider several cases.

\begin{enumerate}
\item $j=i.$ 
\begin{eqnarray*}
UE_{i}^{x^{\prime }y^{\prime }}\left( \mathbf{1}^{ik}\right) &=&\frac{%
n-2nx^{\prime }}{n-2}U_{i}\left( \mathbf{1}^{ik}\right) +\frac{2nx^{\prime
}-2}{n-2}ES_{i}\left( \mathbf{1}^{ik}\right) -\left( x^{\prime }-y^{\prime
}\right) CD_{i}\left( \left( \mathbf{1}^{ik}\right) ^{i0}\right) \\
&=&\frac{n-2nx^{\prime }}{n-2}\frac{1}{n}+\frac{2nx^{\prime }-2}{n-2}\frac{1%
}{2}=x^{\prime }\geq 0.
\end{eqnarray*}

\item $k=i.$ 
\begin{equation*}
UE_{i}^{x^{\prime }y^{\prime }}\left( \mathbf{1}^{ji}\right) =\frac{%
n-2nx^{\prime }}{n-2}\frac{1}{n}+\frac{2nx^{\prime }-2}{n-2}\frac{1}{2}%
-\left( x^{\prime }-y^{\prime }\right) =y^{\prime }\geq 0.
\end{equation*}

\item $i\in N\backslash \left\{ j,k\right\} .$ 
\begin{equation*}
UE_{i}^{x^{\prime }y^{\prime }}\left( \mathbf{1}^{jk}\right) =\frac{%
n-2nx^{\prime }}{n-2}\frac{1}{n}-\left( x^{\prime }-y^{\prime }\right) \frac{%
-1}{n-2}=\frac{1-x^{\prime }-y^{\prime }}{n-2}\geq 0.
\end{equation*}
\end{enumerate}

Conversely, let $R$ be a rule satisfying all the axioms from the statement.
Similarly to the proof of Theorem \ref{char AD+AN+MA}, we can find $%
x^{\prime },y^{\prime }$ such that, for each pair $i,j\in N$ with $i\neq j,$ 
$x^{\prime }=R_{i}\left( \boldsymbol{\mathbf{1}}^{ij}\right) ,$ $y^{\prime
}=R_{j}\left( \boldsymbol{\mathbf{1}}^{ij}\right) ,$ and for each $k\in
N\backslash \left\{ i,j\right\} ,$ $\frac{1-x^{\prime }-y^{\prime }}{n-2}%
=R_{k}\left( \boldsymbol{1}^{ij}\right) $.

By \textit{home order preservation}, 
\begin{equation*}
\frac{1-x^{\prime }-y^{\prime }}{n-2}\leq y^{\prime }\leq x^{\prime }.
\end{equation*}

Then,%
\begin{equation*}
x^{\prime }\geq \frac{1-x^{\prime }-y^{\prime }}{n-2}\geq \frac{1-2x^{\prime
}}{n-2}.
\end{equation*}

Notice that $x^{\prime }=y^{\prime }=\frac{1}{n}$ satisfies the previous
conditions. Then, $x^{\prime }\geq \frac{1}{n}.$

By \textit{non-negativity}, $x^{\prime }\geq 0,$ $y^{\prime }\geq 0$ and 
\begin{equation*}
\frac{1-x^{\prime }-y^{\prime }}{n-2}\geq 0\Leftrightarrow x^{\prime
}+y^{\prime }\leq 1.
\end{equation*}

Then, $x^{\prime }\in \left[ \frac{1}{n},1\right] $ and $y^{\prime }\in %
\left[ 0,1-x^{\prime }\right] .$

Similarly to the proof of part $\left( 1\right) $ of Theorem \ref{char
AD+AN+MA}, we can prove that, for each $i\in N$, 
\begin{equation*}
R_{i}(A)=x^{\prime }\alpha _{i}\left( A\right) +\frac{1-2x^{\prime }}{n-2}%
\left( \left\vert \left\vert A\right\vert \right\vert -\alpha _{i}\left(
A\right) \right) -\left( x^{\prime }-y^{\prime }\right) CD_{i}\left(
A^{i0}\right) .
\end{equation*}%
Now, 
\begin{equation*}
x^{\prime }\alpha _{i}\left( A\right) +\frac{1-2x^{\prime }}{n-2}\left(
\left\vert \left\vert A\right\vert \right\vert -\alpha _{i}\left( A\right)
\right) =\frac{n\left( 1-2x^{\prime }\right) }{n-2}\frac{\left\vert
\left\vert A\right\vert \right\vert }{n}+\frac{2\left( nx^{\prime }-1\right) 
}{n-2}\frac{\alpha _{i}\left( A\right) }{2}.
\end{equation*}

Let 
\begin{equation*}
\lambda =\frac{n\left( 1-2x^{\prime }\right) }{n-2}.
\end{equation*}

As $\frac{1}{n}\leq x^{\prime }\leq 1$, it follows that $\frac{-n}{n-2}\leq
\lambda \leq 1.$ Thus, 
\begin{equation*}
x^{\prime }\alpha _{i}\left( A\right) +\frac{1-2x^{\prime }}{n-2}\left(
\left\vert \left\vert A\right\vert \right\vert -\alpha _{i}\left( A\right)
\right) =UE_{i}^{\lambda }\left( A\right) .
\end{equation*}

$\left( 2\right) $ Using similar arguments to those used in $\left( 1\right) 
$ we can prove that any \textit{extended $UC$ rule} with $x^{\prime }\leq
y^{\prime }$ satisfies \textit{additivity}, \textit{anonymity}, \textit{%
non-negativity} and \textit{away order preservation}. Conversely, let $R$ be
a rule satisfying all those axioms. Then, similarly to the proof of Theorem %
\ref{char AD+AN+MA}, we can find $x^{\prime },y^{\prime }$ such that, for
each pair $i,j\in N$ with $i\neq j,$ $x^{\prime }=R_{i}\left( \boldsymbol{1}%
^{ij}\right) ,$ $y^{\prime }=R_{j}\left( \boldsymbol{1}^{ij}\right) ,$ and
for each $k\in N\backslash \left\{ i,j\right\} ,$ $\frac{1-x^{\prime
}-y^{\prime }}{n-2}=R_{k}\left( \boldsymbol{1}^{ij}\right) $. Similarly to $%
\left( 1\right) $, we can prove that $x^{\prime }\leq y^{\prime },$ $%
y^{\prime }\in \left[ \frac{1}{n},1\right] $, $x^{\prime }\in \left[
0,1-y^{\prime }\right] ,$ $\lambda =\frac{n\left( 1-2y^{\prime }\right) }{n-2%
},$ and for each $i\in N$ 
\begin{equation*}
R_{i}(A)=UE_{i}^{\lambda }\left( A\right) -\left( y^{\prime }-x^{\prime
}\right) CD_{i}\left( A^{0i}\right) ,
\end{equation*}%
as desired.
\end{proof}

\bigskip

Theorem \ref{char AD+AN+WOP+NN} also highlights the central role of \textit{%
concede-and-divide} as a rule to solve broadcasting problems. Extended $UE$
rules have two parts: one depending on the family of $UE$ rules $\left(
UE_{i}^{\lambda }\left( A\right) \right) $ and another depending on \textit{%
concede-and-divide} ($CD_{i}\left( A^{i0}\right) $ and $CD_{i}\left(
A^{0i}\right) ).$

\begin{remark}
\label{indep AD+AN+WOP+NN}$\left( 1\right) $ The axioms used in Theorem \ref%
{char AD+AN+WOP+NN} are independent.

$R^{6},$ defined as in Remark \ref{indep AD+AN+WOP+WUB}, satisfies $AD,$ $%
AN, $ and $HOP$ but fails $NN.$

$R^{7},$ defined as in Remark \ref{indep AD+AN+WOP+WUB}, satisfies $AD,$ $%
AN, $ and $NN$ but fails $HOP.$

Let $N=\left\{ 1,2,3\right\} $ and $R^{11}$ be the separable rule (see
Berganti\~{n}os and Moreno-Ternero, 2022c) where for each pair $i,j\in N$, 
\begin{equation*}
\left( x^{ij}\right) _{i,j\in N}=\left( 
\begin{array}{ccc}
& 0.50 & 0.60 \\ 
0.45 &  & 0.65 \\ 
0.50 & 0.60 & 
\end{array}%
\right) \text{ and }\left( y^{ij}\right) _{i,j\in N}=\left( 
\begin{array}{ccc}
& 0.40 & 0.35 \\ 
0.35 &  & 0.25 \\ 
0.25 & 0.20 & 
\end{array}%
\right)
\end{equation*}

Then, $R^{11}$ that satisfies $AD,$ $HOP,$ and $NN$ but fails $AN$.

For each $A\in \mathcal{P}$, and each $i\in N,$ let 
as%
\begin{equation*}
R_{i}^{10}(A)=\left\{ 
\begin{tabular}{ll}
$U\left( A\right) $ & if $\left\vert \left\vert A\right\vert \right\vert
\leq 10$ \\ 
$ES\left( A\right) $ & if $\left\vert \left\vert A\right\vert \right\vert
>10.$%
\end{tabular}%
\right.
\end{equation*}

Then, $R^{12}$ satisfies $AN,$ $HOP,$ and $NN$ but fails $AD.$ \bigskip

$\left( 2\right) $ $R^{6},$ defined as in Remark \ref{indep AD+AN+WOP+WUB},
satisfies $AD,$ $AN,$ and $AOP$ but fails $NN.$

$R^{7},$ defined as in Remark \ref{indep AD+AN+WOP+WUB}, satisfies $AD,$ $%
AN, $ and $NN$ but fails $AOP.$

Let $N=\left\{ 1,2,3\right\} $ and $R^{13}$ be the separable rule (see
Berganti\~{n}os and Moreno-Ternero, 2022c) where for each $i,j\in N$, 
\begin{equation*}
\left( x^{ij}\right) _{i,j\in N}=\left( 
\begin{array}{ccc}
& 0.40 & 0.35 \\ 
0.35 &  & 0.25 \\ 
0.25 & 0.20 & 
\end{array}%
\right) \text{ and }\left( y^{ij}\right) _{i,j\in N}=\left( 
\begin{array}{ccc}
& 0.50 & 0.60 \\ 
0.45 &  & 0.65 \\ 
0.50 & 0.60 & 
\end{array}%
\right)
\end{equation*}

Then, $R^{13}$ that satisfies $AD,$ $AOP,$ and $NN$ but fails $AN$.

Finally, $R^{10}$, defined as in Remark 3, satisfies $AN,$ $AOP,$ and $NN$
but fails $AD.$
\end{remark}

\bigskip

A trivial consequence of Theorem \ref{char AD+AN+WOP+NN} and part 4 of
Proposition \ref{rel axioms} is that a rule satisfies \textit{additivity}, 
\textit{anonymity}, \textit{order preservation} and \textit{non-negativity}
if and only if it is a $UE$ rule. We do not stress this result because 
\textit{anonymity} is actually redundant for this characterization, as shown
in Berganti\~{n}os and Moreno-Ternero (2022a).

\section{Conclusion}

We have studied in this paper the impact of \textit{anonymity} as an axiom
for \textit{broadcasting problems}.

On the one hand, we have shown that, when combined with some axioms, it is
possible to obtain new characterizations of rules and families of rules
already studied in \textit{broadcasting problems}. For instance, by adding 
\textit{essential team} (as well as \textit{additivity}), we characterize 
\textit{concede-and-divide}. With \textit{null team} instead of \textit{%
essential team}, we characterize the family of \textit{generalized split
rules}, which generalize the focal \textit{equal split rule}.

On the other hand, we have also shown that, when combined with some other
axioms, it is possible to characterize new families of rules: the so called 
\textit{extended }$EC$\textit{\ rules}, \textit{extended }$UC$\textit{\ rules%
}, and \textit{extended }$UE$\textit{\ rules}. All these extended rules can
be described as the sum of two components. In the first one, we apply to the
original problem a rule within the corresponding family of rules already
studied in \textit{broadcasting problems} ($EC$ rules, $UC$ rules or $UE$
rules). In the second component, we always apply the focal \textit{%
concede-and divide} to the resulting problem after nullifying some of
audiences in the original problem.

\bigskip


\bigskip

\bigskip

\end{document}